\documentclass[prl,aps,amsmath, amssymb, twocolumn, superscriptaddress,longbibliography]{revtex4-2}

\usepackage{graphicx}
\usepackage[colorlinks=true, linkcolor=red, urlcolor=blue]{hyperref}
\usepackage{physics}
\usepackage{color}
\usepackage{enumerate}
\usepackage{bm}
\usepackage[normalem]{ulem}
\usepackage{comment}
\usepackage{dsfont}
\usepackage{multirow}
\usepackage{appendix}

\usepackage[caption=false]{subfig}

\usepackage{braket}
\usepackage{amsmath}

\usepackage{amsthm}
\usepackage{amsfonts, amssymb}
\usepackage{bbm}

\newtheorem{theorem}{Theorem}
\newtheorem{remark}{Remark}

\newtheorem{corollary}{Corollary}
\usepackage{mathtools}
\usepackage{mathrsfs}

\newcommand{\ergo}{\mathcal{E}}

\newcommand{\en}{\mathfrak{E}}

\newcommand{\dstate}{\hat{\rho}}

\newcommand{\ham}{\hat{H}}

\newcommand{\pos}{\hat{q}}
\newcommand{\mom}{\hat{p}}

\begin{document}
	
	\title{Work extraction processes from noisy quantum batteries: \\ the role of non local resources } 
 
	\author{Salvatore Tirone}
 \email{salvatore.tirone@sns.it}
	\affiliation{Scuola Normale Superiore, I-56126 Pisa, Italy}  
 
	\author{Raffaele Salvia}
\email{raffaele.salvia@sns.it}
	\affiliation{NEST, Scuola Normale Superiore and Istituto Nanoscienze-CNR, I-56126 Pisa, Italy}
    
	\author{Stefano Chessa}
 \email{schessa@illinois.edu}
	\affiliation{NEST, Scuola Normale Superiore and Istituto Nanoscienze-CNR, I-56126 Pisa, Italy}
	\affiliation{Electrical and Computer Engineering, University of Illinois Urbana-Champaign, Urbana, Illinois, 61801, USA}
    
	\author{Vittorio Giovannetti}
	\affiliation{NEST, Scuola Normale Superiore and Istituto Nanoscienze-CNR, I-56126 Pisa, Italy}
	
	\begin{abstract}
We demonstrate an asymmetry 
between the beneficial effects one can obtain using non-local operations and non-local states
to mitigate the detrimental effects of environmental noise in the work extraction from quantum battery models. 
Specifically, we show that using non-local 
recovery operations after the noise action can in general increase the amount of work one can recover from the battery even with separable (i.e. non entangled) input states. On the contrary, 
employing entangled input states with local recovery operations will not generally improve the battery performances. 
	\end{abstract}
	
	\maketitle

Quantum thermodynamics is a rapidly growing field that seeks to understand the behavior of small quantum systems at the nanoscale \cite{Goold2016}. Of particular interest is the use of quantum effects to improve the charging processes of batteries \cite{Alicki2013, Campaioli2018, Andolina2019, Farina2019, Rossini2019, rossini2019quantum, JuliFarr2020}, which could potentially lead to technological advancements in a variety of sectors. 
A crucial aspect of the problem is to assess the stability of these models when in contact with  environmental noise, as this represents a more realistic scenario and it may have a significant impact on the efficiency of the energy recovery/storage.
In noiseless regimes,  capacities for quantum battery models  have been proposed in~\cite{def_ergo, JuliFarr2020, battery_cap2023} while, from a resource theoretical point of view, the thermodynamic capacity (in the sense of simulability) of quantum channels has been defined \cite{Thermodyn_Cap2019}. Concerning energy manipulation in more realistic scenarios where environmental noise is acting on the system,
 a few results have been obtained in specialized settings, see e.g.~\cite{Carrega_2020, Bai_2020, Tabesh_2020, Ghosh_2021, Santos_2021, Zakavati_2021, Landi_2021, Morrone_2022, Sen_2023}, 
and various schemes have been proposed to stabilize quantum batteries in the presence of specific type of pertubations \cite{Liu2019, PhysRevE.100.032107, PhysRevApplied.14.024092, PhysRevA.102.060201, PhysRevE.101.062114, PhysRevResearch.2.013095, Goold2016, liu2021boosting, PhysRevE.103.042118, PhysRevE.105.054115}. In this work we tackle the problem using the \textit{work capacitance} functionals introduced in Ref.~\cite{quantumworkcapacitances}. These quantities gauge the efficiency of the work extraction process from quantum battery models formed by large collections of identical and independent noisy elements (quantum cells or q-cells in brief), targeting  optimal state preparation schemes  (encoding operations performed \textit{before} the action of noise) and optimal recovery transformations (decoding operations performed \textit{after} the noise). Formally  they are defined as the asymptotic limit of the ratio between the work extracted from the system and the initial energy stored in the quantum battery and, for a given quantum battery model, their  values strongly depend on the type of constraints one enforces on the transformations allowed on the system -- see Fig.~\ref{fig:scheme1}.
Our main finding is to provide evidence that, irrespective of the noise model, the mere use of non-local resources at the level of the encoding (i.e.  employing  quantum correlated states to store the initial  energy of the battery) does not improve the efficiency of the model. On the contrary we show that employing non local transformations at the recovering stage, will in general increase the work extraction performance of the battery. The key ingredient to attain such result is the derivation of a single-letter formula that allows us to simplify the evaluation of the local-ergotropy capacitances~\cite{quantumworkcapacitances} for arbitrary noise models.

\begin{figure}[h]
\centering
\includegraphics[width=0.48\textwidth]{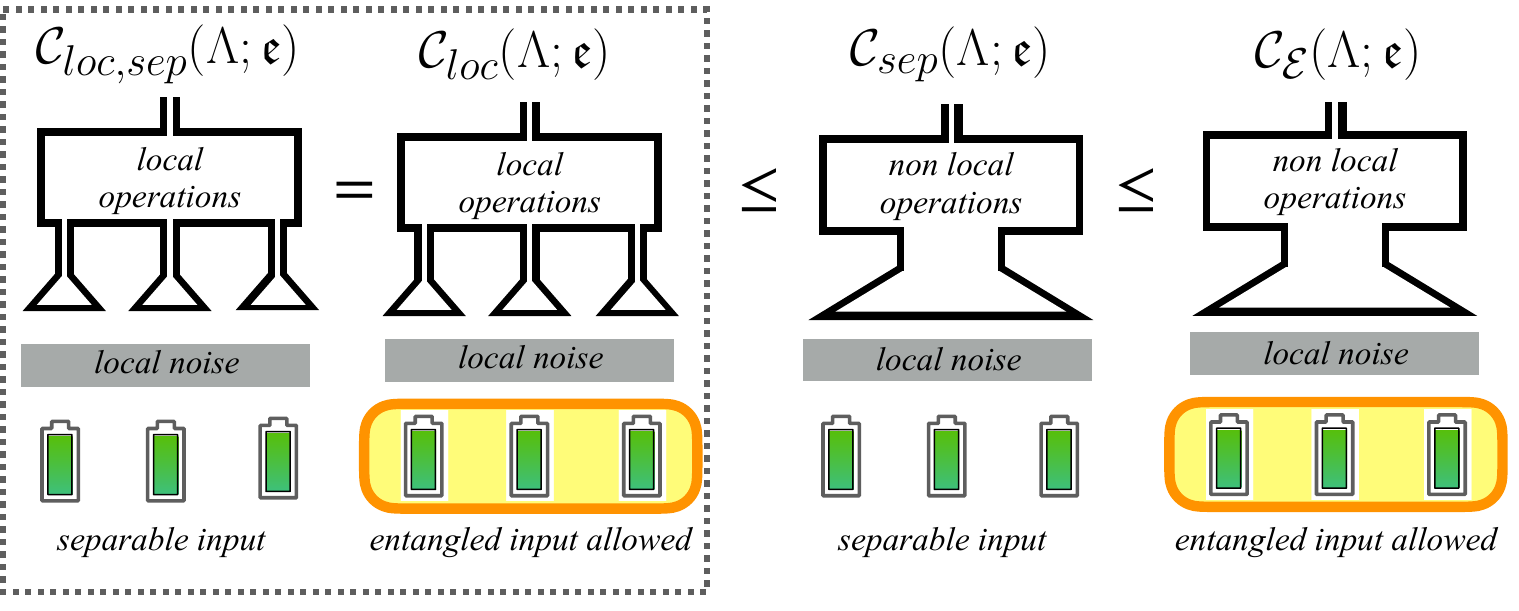}
\caption{Resource accounting for  work extraction in noisy QB models composed by $n$ q-cells (green elements of the figure) affected by local noise (grey). From left to right:
separable-input, local ergotropy capacitance
$C_{\rm loc, sep} \left( \Lambda; \mathfrak{e} \right)$ (maximum work extractable per unit cell,
when both the state preparation of the QB and the recovery operations applied after the noise action, are restricted to local resources);  local ergotropy capacitance
$C_{\rm loc} \left( \Lambda; \mathfrak{e} \right)$ (here locality is enforced only at the recovery level); separable-input  capacitance
$C_{\rm sep} \left( \Lambda; \mathfrak{e} \right)$ (locality is enforced only at the 
 state preparation level); ergotropy capacitance
$C_{\cal E} \left( \Lambda; \mathfrak{e} \right)$ (no local restrictions imposed). 
As shown in Theorem~\ref{th1}, irrespective of the noise model, $C_{\rm loc} \left( \Lambda; \mathfrak{e} \right)$ and  $C_{\rm loc, sep} \left( \Lambda; \mathfrak{e} \right)$ always coincide.}
\label{fig:scheme1}
\end{figure}

\paragraph{Preliminaries:}
We model noisy Quantum Batteries (QBs) as collections of  $n$ identical  and independent (non-interacting) elements (quantum cells or q-cells in brief), each capable to store energy in the internal degrees of freedom associated with their local Hamiltonians $\hat{h}_1$, $\hat{h}_2$, $\cdots$, $\hat{h}_n$, all of which are locally identical to a single q-cell Hamiltonian $\hat{h}$, and perturbed by the same local noise source which we describe in terms of a linear, completely positive and trace preserving (LCTP) super-operator $\Lambda$~\cite{STINE, CHOI1975285, kraus1983states}.  
The work extraction efficiency in these systems can be measured using as figure of merit the work-capacitances introduced in~\cite{quantumworkcapacitances}.
A first example  is provided by the ergotropic capacitance $C_{\ergo}\left( \Lambda; \mathfrak{e} \right)$.  Ergotropy is a well established measure of the maximum work one can extract from a quantum  state $\dstate$ by means of unitary operations that preserve the total energy of the system~\cite{def_ergo, Niedenzu2019, obejko2021}.
For a $d$-dimensional system characterized by a Hamiltonian $\ham$,
it can be expressed as  
\begin{eqnarray}\label{ergoeeeDEFinvariance} 
\ergo(\dstate;\ham)  := \max_{\hat{U} \in {\mathbb U}(d)}\Big\{  \en(\dstate;\ham) 
-\en(\hat{U} \dstate \hat{U}^\dag;\ham)\Big\}   \;,
\end{eqnarray}
where $\en(\dstate ; \ham) := \Tr[\dstate\ham]$ is the average energy of a quantum state $\dstate$ and 
 $\mathbb{U}(d)$ is the $d$-dimensional representation of the unitary group.
 In view of this definition a reasonable way to gauge the maximum work one  can retrieve form the QB after the action of the noise $\Lambda$ is 
 obtained by considering \begin{eqnarray} 
 \ergo^{(n)}(\Lambda;E)&:=& \max_{\dstate^{(n)} \in \mathfrak{S}^{(n)}_{E}}
 \ergo(\Lambda^{\otimes n}(\dstate^{(n)});\ham^{(n)}) \;,   \label{dergo}
 \end{eqnarray} where $\ham^{(n)}:=\hat{h}_1 + \cdots + \hat{h}_n$ represents the battery Hamiltonian, and where the maximization is performed on the set $\mathfrak{S}^{(n)}_{E}$, which is the set of all the $n$ q-cells states $\dstate^{(n)}$ with average energy $\en(\dstate^{(n)})\leq E$.
 The ergotropic capacitance $C_{\ergo}\left( \Lambda; \mathfrak{e} \right)$ is now defined 
as a proper regularization of  $\ergo^{(n)}(\Lambda;E)$ for $n\rightarrow \infty$,  under the assumption that 
(on average) each of the $n$ q-cells stores no more than a fraction $\mathfrak{e} \in [0,||\hat{h}||_{\infty}]$ of the total input energy~\cite{quantumworkcapacitances}, i.e. 
\begin{eqnarray} 
C_{\ergo}\left( \Lambda; \mathfrak{e} \right) := \lim_{n\rightarrow \infty} 
\frac{\ergo^{(n)}(\Lambda;E=n\mathfrak{e})}{n} \;, \label{cergo} 
\end{eqnarray} 
where without loss of generality we set to zero the ground state energy of  the local Hamiltonians $\hat{h}$. In absence of noise (i. e. $\Lambda$ is the identity channel) $C_{\ergo}\left( \Lambda; \mathfrak{e} \right)$ is equal to $\mathfrak{e}$, signalling that, by properly preparing the input state of the q-cells, we can retrieve all the energy we have initially stored into the battery. Dissipation and decoherence  will tend instead to produce smaller values of $C_{\ergo}\left( \Lambda, \mathfrak{e} \right)$, indicating that the performance of the model gets degraded irrespective of the choice we make at the level of state preparation of the QB. Setting restrictions on the allowed operations one can perform on the battery will also reduce the value of $C_{\ergo}\left( \Lambda, \mathfrak{e} \right)$. For instance, assuming the maximization in Eq.~(\ref{dergo}) to run only on separable input states of the q-cells will lead us to replace $\ergo^{(n)} (\Lambda;n\mathfrak{e})$ with  $\ergo^{(n)}_{\rm sep} (\Lambda;n\mathfrak{e})(\leq \ergo^{(n)} (\Lambda;n\mathfrak{e}))$, which regularized as in~(\ref{cergo}) gives the separable-input  capacitance
$C_{\rm sep} \left( \Lambda, \mathfrak{e} \right)$ of the model. This, in turn, expresses the asymptotic work we can extract per q-cell in the absence of initial entanglement between these elements~\cite{quantumworkcapacitances}. Similarly, by restricting the optimization in Eq.~(\ref{ergoeeeDEFinvariance}) to include only unitary operations acting locally on the q-cells (i.e. replacing $\ergo(\dstate;\ham)$ with the local ergotropy \cite{salviadepalma}), will lead us to identify $\ergo^{(n)}_{\rm loc} (\Lambda;n\mathfrak{e})$ with its regularized limit 
$C_{\rm loc} \left( \Lambda; \mathfrak{e} \right)$. Finally, assuming the optimizations  to be restricted to separable states and to local unitary operations, one can define $\ergo^{(n)}_{\rm loc, sep} (\Lambda;n\mathfrak{e})$ and  $C_{\rm loc, sep} \left( \Lambda; \mathfrak{e} \right)$, See \cite{SM} for the formal definitions. Simple resource counting arguments can be used to show that a natural partial ordering exists among these functionals~\cite{quantumworkcapacitances} which identifies $C_{\cal E} \left( \Lambda; \mathfrak{e} \right)$ and $C_{\rm loc, sep} \left( \Lambda; \mathfrak{e} \right)$ as the largest and smallest terms respectively, leaving to $C_{\rm sep} \left( \Lambda; \mathfrak{e} \right)$ and $C_{\rm loc} \left( \Lambda; \mathfrak{e} \right)$ the role of intermediate quantities,  i.e. 
 \begin{eqnarray}\label{natural}  C_{\cal E} \left( \Lambda; \mathfrak{e} \right)\geq C_{\rm sep} \left( \Lambda; \mathfrak{e} \right),
C_{\rm loc} \left( \Lambda; \mathfrak{e} \right)\geq  C_{\rm loc, sep} \left( \Lambda; \mathfrak{e} \right)\;. 
  \end{eqnarray} 
As we shall see, one of the main results of the present work is to refine~(\ref{natural}) showing that, for all noise models, no gap exists between $C_{\rm loc, sep} \left( \Lambda; \mathfrak{e} \right)$ and $C_{\rm loc} \left( \Lambda; \mathfrak{e} \right)$, and that provably $C_{\rm sep} \left( \Lambda; \mathfrak{e} \right) \geq C_{\rm loc} \left( \Lambda; \mathfrak{e} \right)$.

\paragraph{Closed formulas and bounds:--} 
We now show that, regardless of the LCPT map $\Lambda$, the separable-input, local ergotropy capacitance
$C_{\rm loc, sep} \left( \Lambda; \mathfrak{e} \right)$ and 
the local ergotropy capacitance
$C_{\rm loc} \left( \Lambda, \mathfrak{e} \right)$ coincide and admit a simple single-letter expression in terms of the single-shot ($n=1$)
maximal output ergotropy functional~(\ref{dergo}), i.e. 

\begin{theorem} \label{th1} For any LCPT map $\Lambda$, and for any $\mathfrak{e}\in [0,||\hat{h}||_{\infty}]$ we can write 
\begin{eqnarray} \label{exact} 
C_{\rm loc} \left( \Lambda; \mathfrak{e} \right)&=&
C_{\rm loc, sep} \left( \Lambda; \mathfrak{e} \right)= \chi(\Lambda; \mathfrak{e}) \;, \\ 
\chi(\Lambda; \mathfrak{e})&:=&  \sup_{\{ p_j, \mathfrak{e}_j\}} \sum_j p_j\;  \ergo^{(1)}(\Lambda;\mathfrak{e}_j)\;, 
\end{eqnarray} 
where the supremum is taken over all the distributions $\{ p_j, \mathfrak{e}_j \}$ of input q-cell energy $\mathfrak{e}_j\in [0,||\hat{h}||_{\infty}]$ that fulfil the
constraint
\begin{eqnarray}\label{energyconstraint} 
\sum_j p_j  \mathfrak{e}_j\leq	 \mathfrak{e} \;. 
\end{eqnarray} 
\end{theorem} 
\begin{proof} 
In view of (\ref{natural}), to derive Eq.~(\ref{exact}) it is sufficient to show that $\chi(\Lambda; \mathfrak{e})$ 
 is {\it i)} a lower bound
for $C_{\rm loc, sep} \left( \Lambda; \mathfrak{e} \right)$, and {\it ii)} an upper bound for 
$C_{\rm loc} \left( \Lambda; \mathfrak{e} \right)$.

A proof  of the  inequality {\it i)} follows by observing that 
 given $\mathfrak{e}\in [0,||\hat{h}||_{\infty}]$ and $n$ integer, we have that the maximum output local ergotropy of the model 
 $\ergo_{\rm loc, sep}^{(n)}(\Lambda;E=n\mathfrak{e})$ is certainly smaller than the output local ergotropy  computed on the (factorized) state
 of the form 
 $\hat{\rho}_{\rm fact}^{(n)} := \hat{\rho}_{1} \otimes \hat{\rho}_{2}\otimes \cdots\otimes \hat{\rho}_{n}$, 
 where for $i=1,\cdots,n$, $\hat{\rho}_{i}$  is a density matrix of the $i$-th q-cell with input energy $\tilde{\mathfrak{e}}_i\in [0,||\hat{h}||_{\infty}]$ fulfilling the constraint
 $\sum_{i=1}^n \tilde{\mathfrak{e}}_i /n\leq  \mathfrak{e}$.  
 That is 
 $\ergo_{\rm loc, sep}^{(n)}(\Lambda;E=n\mathfrak{e}) \geq \ergo_{\rm loc} (\Lambda^{\otimes n}( \hat{\rho}_{\rm fact}^{(n)});\ham^{(n)}) 
 =\sum_{i=1}^n  \ergo(\Lambda( \hat{\rho}_i);\hat{h})$,
 where in the second passage we used the fact that the local ergotropy for non-interacting systems is additive
 ~\cite{quantumworkcapacitances,salviadepalma}.
 In particular, selecting the $\tilde{\rho}_{i}$ so that they maximize  the single-shot 
 maximal output ergotropy $\ergo^{(1)}(\Lambda; \tilde{\mathfrak{e}}_i)$, we can translate the above inequality into
\begin{eqnarray} \label{questo1} 
 \frac{\ergo_{\rm loc, sep}^{(n)}(\Lambda;E=n\mathfrak{e})}{n} &\geq& 
\sum_{i=1}^n \frac{\ergo^{(1)}(\Lambda; \tilde{\mathfrak{e}}_i)}{n}  \;. 
 \end{eqnarray} 
 In the $n\rightarrow \infty$ limit the l.h.s. converges toward $C_{\rm loc, sep} \left( \Lambda, \mathfrak{e} \right)$. 
 On the contrary, given an arbitrary distribution $\{ p_j, \mathfrak{e}_j\}$ that fulfils the constraint (\ref{energyconstraint}), we can
 force the r.h.s. of (\ref{questo1}) to converge to $\sum_j p_j\;  \ergo^{(1)}(\Lambda;\mathfrak{e}_j)$. 
  Accordingly we can write
$C_{\rm loc, sep} \left( \Lambda; \mathfrak{e} \right)\geq   \sum_j p_j\;  \ergo^{(1)}(\Lambda;\mathfrak{e}_j)$, 
 which upon optimization over all choices of $\{ p_j, \mathfrak{e}_j\}$ shows that
 indeed   the r.h.s. of Eq.~(\ref{exact}) is a lower bound for
 $C_{\rm loc, sep} \left( \Lambda; \mathfrak{e} \right)$.
 
\noindent We now prove property {\it ii)}. For this purpose 
 consider the optimal  state $\dstate_*^{(n)}$ which allows us to
 saturate the maximization of $\ergo_{\rm loc}^{(n)}(\Lambda;E=n\mathfrak{e})$ for fixed $n$ and $\mathfrak{e}$.
 Using the additivity of  the local ergotropy we can write 
 \begin{eqnarray} \label{ddfds} 
 \frac{\ergo_{\rm loc}^{(n)}(\Lambda;E=n\mathfrak{e})}{n} &=& \frac{\ergo_{\rm loc} (\Lambda^{\otimes n}( \hat{\rho}_*^{(n)});\ham^{(n)})}{n}  \\ &=& \sum_{i=1}^n \frac{\ergo (\Lambda( \hat{\rho}_i); \hat{h})}{n} \leq 
 \sum_{i=1}^n \frac{\ergo^{(1)}(\Lambda; {\mathfrak{e}}_i)}{n}\;, 
 \nonumber 
 \end{eqnarray} 
where for $i\in\{ 1,\cdots,n\}$,  $\hat{\rho}_i$ is the reduced density matrix of $\hat{\rho}_*^{(n)}$ associated
with the $i$-th q-cell and where $\mathfrak{e}_i$ indicates its mean energy, which by construction must fulfil the 
condition 
 $\sum_{i=1}^n  {{\mathfrak{e}}_i}/{n} \leq  \mathfrak{e}$. 
 Observe next that since $\{ p_i=1/n,  \mathfrak{e}_i\}$ is a special 
 instance of energy distribution satisfying (\ref{energyconstraint}),
 the last term of~(\ref{ddfds}) is certainly smaller than or equal to  $\chi(\Lambda; \mathfrak{e})$. Taking the $n\rightarrow \infty$ limit we finally arrive to the thesis.
 \end{proof}

\begin{remark} 
For noise models where the single-shot ($n=1$)
maximal output ergotropy $\ergo^{(1)}(\Lambda;\mathfrak{e})$ is a concave function of the energy parameter $\mathfrak{e}$
the optimization in Eq.~(\ref{exact}) can be
explicitly performed leading to a more compact expression:
\begin{eqnarray} \label{exactconcave} 
C_{\rm loc} \left( \Lambda; \mathfrak{e} \right)=
C_{\rm loc, sep} \left( \Lambda; \mathfrak{e} \right)= \ergo^{(1)}(\Lambda;\mathfrak{e})\;.
\end{eqnarray} 
\end{remark} 
 
Since $C_{\rm loc, sep} \left( \Lambda; \mathfrak{e} \right)$ is certainly not larger than 
 $C_{\rm sep} \left( \Lambda; \mathfrak{e} \right)$, the result of Theorem~\ref{th1} allows us to introduce
 a definite ordering among $C_{\rm sep} \left( \Lambda; \mathfrak{e} \right)$ and $C_{\rm loc} \left( \Lambda; \mathfrak{e} \right)$, i.e. 
\begin{eqnarray} \label{ordering} 
C_{\rm sep} \left( \Lambda; \mathfrak{e} \right) \geq 
C_{\rm loc} \left( \Lambda; \mathfrak{e} \right)\;. \end{eqnarray} 
Physically this implies that, at variance with what happens in other quantum information settings like those of quantum communication~\cite{HOLEVOBOOK}, in the case of work extraction tasks the use of non local resources at the decoding stage is, in principle, {\it always} preferable than their use at the encoding stage. In order to strengthen this
statement we  next derive a non-trivial lower bound for $C_{\rm sep} \left( \Lambda; \mathfrak{e} \right)$ which for QB models made of
q-cells of dimension
larger than 2 is typically larger than $\chi(\Lambda, \mathfrak{e})$: 

\begin{corollary}\label{cor1}  For any LCPT map $\Lambda$, and for any $\mathfrak{e}\in [0,||\hat{h}||_{\infty}]$ we can write 
\begin{eqnarray} \label{lowerforcsep} 
C_{\rm sep} \left( \Lambda; \mathfrak{e} \right)\geq 
\chi_{\rm tot}(\Lambda; \mathfrak{e}) :=  \sup_{\{ p_j, \mathfrak{e}_j\}}
 \sum_j p_j\;  \ergo_{\rm tot}^{(1)}(\Lambda;\mathfrak{e}_j)\;, 
\end{eqnarray} 
where $\ergo_{\rm tot}^{(1)}(\Lambda;\mathfrak{e})$ is the single-shot, energy constrained, maximum total ergotropy~\cite{def_ergo} that one can get at the output of the channel $\Lambda$, where
as in the case of (\ref{exact}) the supremum is taken over all the distributions $\{ p_j, \mathfrak{e}_j \}$ of input q-cell energies $\mathfrak{e}_j\in [0,||\hat{h}||_{\infty}]$ satisfying the
constraint~(\ref{energyconstraint}). 
\end{corollary} 
Ultimately the inequality~(\ref{lowerforcsep}) is a consequence of the property shown in Ref.~\cite{quantumworkcapacitances} 
 that 
$C_{\rm sep} \left( \Lambda; \mathfrak{e} \right)$ coincides with $C_{\rm sep,tot} \left( \Lambda; \mathfrak{e} \right)$
(the latter being obtained by replacing the ergotropy appearing in $\ergo_{\rm sep}^{(n)}(\Lambda;E)$ with the total ergotropy): for the sake of completeness however  we provide an independet proof of the corollary in \cite{SM}. 

Under special circumstances we can show that the r.h.s. of Eq.~(\ref{lowerforcsep}) provides  the exact value of $C_{\rm sep} \left( \Lambda; \mathfrak{e} \right)$:
\begin{corollary}\label{cor2}  Suppose that $\Lambda$ is a LCPT map such that,  for all  $\mathfrak{e} \in [0,||\hat{h}||_{\infty}]$  and $n$ integer, it
admits a single-site  state $\hat{\sigma}_{\mathfrak{e}}$, possibly dependent on $n$, with mean energy $\en(\hat{\sigma};\hat{h}) \leq \mathfrak{e}$, 
such that 
\begin{eqnarray}\label{step1} 
\ergo(\Lambda^{\otimes n}(|\Psi_{\rm fact}^{(n)}\rangle\!\langle \Psi_{\rm fact}^{(n)}|);\ham^{(n)})\leq
\ergo((\Lambda(\hat{\sigma}_{\mathfrak{e}})) ^{\otimes n};\ham^{(n)}), 
\end{eqnarray} 
for all factorized pure states $|\Psi_{\rm fact}^{(n)}\rangle$ with mean energy $\en(|\Psi_{\rm fact}^{(n)}\rangle;\ham^{(n)}) \leq n\mathfrak{e}$. Then the inequality~(\ref{lowerforcsep}) is saturated, i.e. \begin{eqnarray} \label{lowerforcsepide} 
C_{\rm sep} \left( \Lambda; \mathfrak{e} \right)=  \chi_{\rm tot}(\Lambda, \mathfrak{e})\;.
\end{eqnarray} 
\end{corollary} 
\begin{proof}
Let be $\{ P_k, |\Psi_{\rm fact}^{(n)}(k)\rangle \}$ an ensemble of factorized states allowing us to express a given 
 separable density matrix  $\dstate_{\rm sep}^{(n)}$ of the QB, i.e. 
$\dstate_{\rm sep}^{(n)} = \sum_k P_k |\Psi_{\rm fact}^{(n)}(k)\rangle\!\langle \Psi_{\rm fact}^{(n)}(k)|$.
 Notice that if $\dstate_{\rm sep}^{(n)}$ has mean energy $\en(\dstate^{(n)}_{\rm sep};\ham^{(n)}) \leq n\mathfrak{e}$, then
 we must have 
$\sum_k P_k \mathfrak{e}_k \leq \mathfrak{e}$,
with $n\mathfrak{e}_k$ being the mean energy of $|\Psi_{\rm fact}^{(n)}(k)\rangle$, i.e.
 $\en(|\Psi_{\rm fact}^{(n)}(k)\rangle;\ham^{(n)}) = n\mathfrak{e}_k$.
Thanks to the convexity of the ergotropy functional, it follows
\begin{eqnarray}\label{step3} 
\frac{1}{n}\ergo(\Lambda^{\otimes n}(\dstate_{\rm sep}^{(n)});\ham^{(n)}) &\leq&\frac{1}{n} \sum_k P_k 
\ergo(\Lambda^{\otimes n}(|\Psi_{\rm fact}^{(n)}(k)\rangle);\ham^{(n)}) \nonumber \\
&\leq& \frac{1}{n} \sum_k P_k  \ergo((\Lambda(\hat{\sigma}_{\mathfrak{e}_k})) ^{\otimes n};\ham^{(n)})\nonumber \\
&\leq& \sum_k P_k  \sup_{n'}\tfrac{ \ergo((\Lambda(\hat{\sigma}_{\mathfrak{e}_k})) ^{\otimes n'};\ham^{(n')})}{n'} \nonumber \\
&=& \sum_k P_k  \ergo_{\rm tot}(\Lambda(\hat{\sigma}_{\mathfrak{e}_k});\hat{h} ) \nonumber \\
&\leq & \sum_k P_k  \ergo^{(1)}_{\rm tot}(\Lambda,{\mathfrak{e}_k})\leq \chi_{\rm tot} (\Lambda;\mathfrak{e})\;,  \nonumber 
\end{eqnarray} 
where in the second inequality we used~(\ref{step1}),  while in the last three passages we invoked  the   definitions of total ergotropy and of
maximum energy-constrained, output total ergotropy. 
Observing that $\{ P_k, {\mathfrak{e}_k}\}$ is a special instance of the ensembles entering the supremum of r.h.s. of Eq.~(\ref{lowerforcsepide}), we can hence conclude that 
$\frac{1}{n}\ergo(\Lambda^{\otimes n}(\dstate_{\rm sep}^{(n)});\ham^{(n)}) \leq  
\chi_{\rm tot} (\Lambda;\mathfrak{e})$,  
that holds true for all separable inputs that have mean energy smaller than or equal to $n\mathfrak{e}$.
The thesis finally follows by taking the supremum with respect all  $\dstate_{\rm sep}^{(n)}$ and taking the $n\rightarrow \infty$ limit. 
\end{proof} 
\paragraph{Examples:--} Since the total-ergotropy functional and the ergotropy always coincide for 
systems of dimension 2 (see \cite{SM}), in view of  Theorem~\ref{th1} and of Corollary~\ref{cor1} the best option to 
identify QBs that exhibit a finite gap $\Delta C$ between $C_{\rm sep} \left( \Lambda; \mathfrak{e} \right)$ and 
$C_{\rm loc} \left( \Lambda; \mathfrak{e} \right)$ is to focus on models with q-cell elements having dimension $d\geq 3$.
A first example of this kind can be found in \cite{quantumworkcapacitances} which solved the values of $C_{\rm sep} \left( \Lambda; \mathfrak{e} \right)$
and $C_{\rm loc} \left( \Lambda; \mathfrak{e} \right)$  for 
 depolarizing maps  of arbitrary dimension. In what follows we present a couple of extra cases. 
 In panel a) of Fig.~\ref{fig:MADworld} for instance we plot the gap we obtained solving numerically the optimizations in Eqs.~(\ref{exact}) and (\ref{lowerforcsep})
  for a three-level system subject to the action of a multilevel amplitude damping channels~\cite{MAD,ReMAD}.
\begin{figure}
\centering
\includegraphics[width=0.95\linewidth]{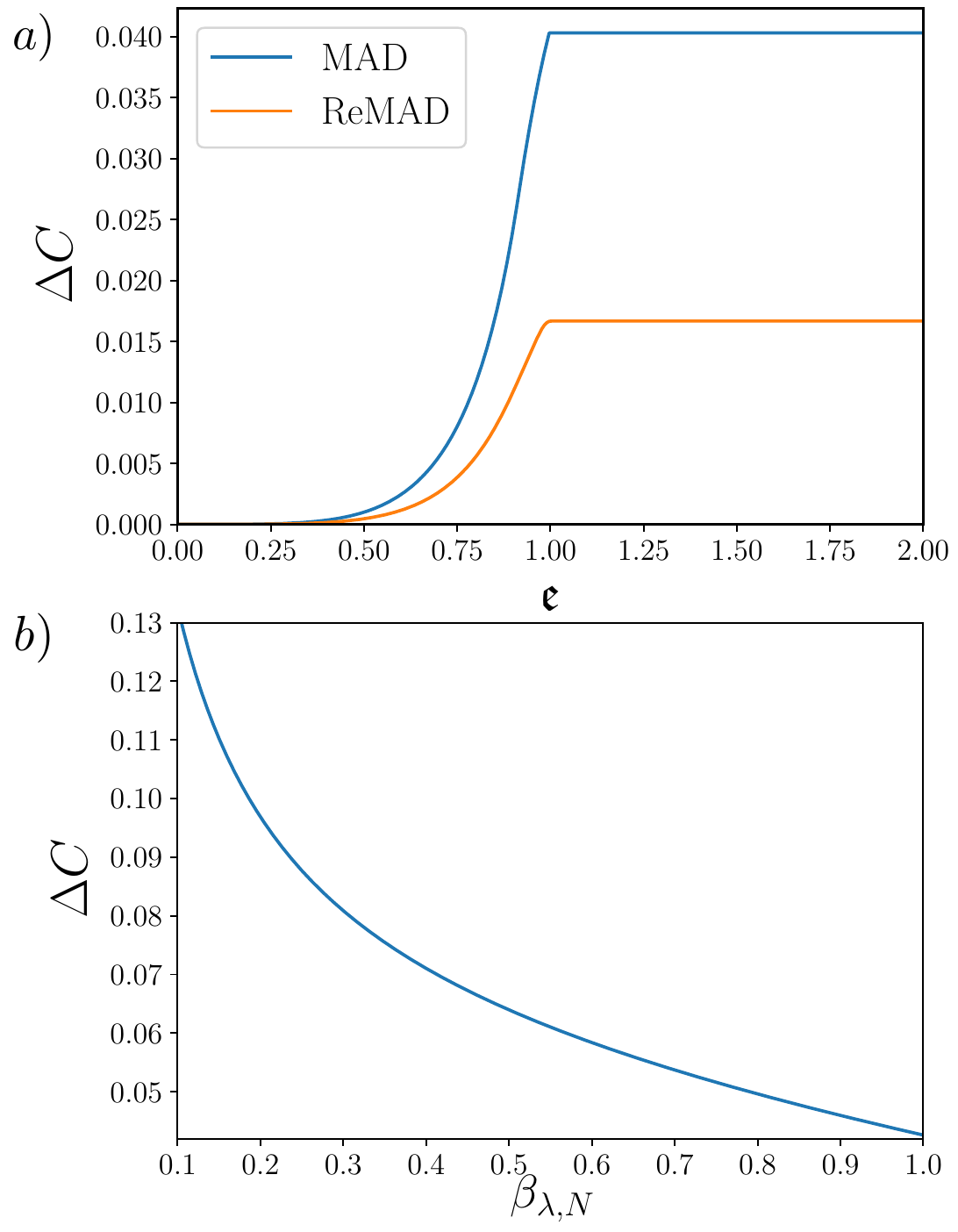}
\caption{Capacitances gap $\Delta C :=  C_{\rm sep}(\Lambda; \mathfrak{e}) - C_{\rm loc}(\Lambda; \mathfrak{e})$. Panel a) plot of the lower bound on $\Delta C$ derived from~(\ref{exact}) and (\ref{lowerforcsep})
 for  the MAD channel $\Phi_{\gamma_1, \gamma_2, \gamma_3}$~\cite{MAD}  and of the 
 ReMAD channel $\Gamma_{\gamma_1, \gamma_2, \gamma_3}$~\cite{ReMAD}, both acting on a qutrit system of Hamiltonian $\hat{h} = \varepsilon_0\left(\ket{1}\!\!\bra{1} + 2\ket{2}\!\!\bra{2}\right)$ 
Here the energy is rescaled by the factor $\varepsilon_0$ and the noise parameters have been fixed equal to $\gamma_1 = 0.3$, $\gamma_2 = 0.2$ and $\gamma_3 = 0.6$. 
  Panel b)  plot of $\Delta C$ 
   for the two-mode Gaussian attenuator $\Phi_2 = \mathcal{L}_{\lambda,N}\otimes\mathcal{L}_{0,0}$ as a function of the effective temperature $\beta_{\lambda, N}$, see Eq.~(\ref{output_coerenti}) in \cite{SM}; in this case $\Delta C$ is independent on the average input energy $\mathfrak{e}$. In the plot the value of $\Delta C$ has been rescaled by $\hbar \omega$, with $\omega$ being the modes frequency.}
\label{fig:MADworld}
\end{figure}
Another example  is obtained by focusing on QB models formed by collection of  independent (infinite dimensional) harmonic oscillators affected by 
 noise models described by Phase-Insensitive Bosonic Gaussian Channels (PI-BGCs)~\cite{HOLEVOBOOK,SERAFINIBOOK}. Thanks to the results of \cite{PhysRevLett.127.210601} we know that in these systems the maximal output ergotropy and maximal total-ergotropy are obtained
 using coherent states as input configurations for the q-cells.   In case $\Lambda$ corresponds to  a single-mode PI-BGC $\Phi_1$ this implies that 
$C_{\ergo}(\Phi_1;\mathfrak{e}) = C_{\rm sep}(\Phi_1;\mathfrak{e})$ for all input energies $\mathfrak{e}$, leading to no gap between
$C_{\rm sep} \left( \Phi_1; \mathfrak{e} \right)$ and $C_{\rm loc} \left( \Phi_1; \mathfrak{e} \right)$. In particular we have  $C_{\ergo}(\mathcal{L}_{\lambda,N};\mathfrak{e}) = \lambda \mathfrak{e}$, $C_{\ergo}(\mathcal{A}_{\mu,N};\mathfrak{e}) = \mu \mathfrak{e}$, and   $C_{\ergo}(\mathcal{N}_{N};\mathfrak{e}) = \mathfrak{e}$, where $\mathcal{L}_{\lambda,N}$ is the thermal attenuator, $\mathcal{A}_{\mu,N}$ is the thermal amplifier and $\mathcal{N}_N$ is the additive noise channel~\cite{HOLEVOBOOK,SERAFINIBOOK}.
 The situation becomes more interesting if $\Lambda$ represents a multi-mode PI-BGC.  As a matter of fact we can observe here
 cases where $C_{\rm sep}$ and $C_{\rm loc}$ exhibit a gap. An instance of this fact is presented in panel b) of Fig.~\ref{fig:MADworld} where we
  consider a 2 modes  PI-BGC  $\Phi_2 = \mathcal{L}_{\lambda,N}\otimes\mathcal{L}_{0,0}$, which acts as a thermal attenuator on one of the modes, and outputs the vacuum state on the other mode.
For this channel the optimal output ergotropy $\ergo^{(1)}\left( \Phi_2; \mathfrak{e} \right)$ and the optimal output total ergotropy $\ergo_{\rm tot}^{(1)}\left( \Phi_2; \mathfrak{e} \right)$ are both affine, but different functions of the input energy $\mathfrak{e}$, implying $C_{\rm loc}\left( \Phi_2; \mathfrak{e} \right) = \ergo^{(1)}\left( \Phi; \mathfrak{e} \right)$ and $C_{\rm sep}\left( \Phi_2; \mathfrak{e} \right) = \ergo^{(1)}_{\rm tot}\left( \Phi; \mathfrak{e} \right)$. As shown in \cite{SM} the resulting gap $\Delta C$ is a constant w.r.t. to the input energy and only depends on the noise parameters of the model (i.e. the constants $\lambda$ and $N$) via an implicit functional of $\lambda$ and~$N$. 
\paragraph{Discussion:--} 
We have observed an asymmetry in the role of non-local resources in the investigation of work extraction from noisy QB models. Specifically, our findings indicate that when energy is recovered through local operations, the use of entangled input states of the q-cells does not enhance the QB's resistance to noise. Conversely, through the examination of a few examples, we demonstrate that incorporating non-locality in the extraction operations can be advantageous. This implies that non-locality provides a distinct benefit in energy retrieval, but may not be beneficial if employed in state preparation. Moving forward, we plan to extend these results to models where the QB Hamiltonian exhibits interactions among the individual q-cells.
\\

We acknowledge financial support by MUR (Ministero dell’ Universit\`a e della Ricerca) through the following projects: 
PNRR MUR project PE0000023-NQSTI, PRIN 2017 Taming complexity via Quantum Strategies: a Hybrid Integrated Photonic approach (QUSHIP) Id. 2017SRN-BRK, and project PRO3 Quantum Pathfinder. S.C. is also supported by a grant through the IBM-Illinois Discovery Accelerator Institute.

\clearpage

\appendix

\onecolumngrid
\begin{center}
    {\bf SUPPLEMENTAL MATERIAL}
\bigskip
\bigskip
\end{center}

\twocolumngrid

\section{Ergotropy and total ergotropy} 
We recall that the ergotropy~(\ref{ergoeeeDEFinvariance}) of quantum state $\dstate$ can be expressed as the difference between its mean energy
and the mean energy of its passive counterpart
$\dstate_{\text{pass}}$, i.e. 
\begin{equation}\label{ergoeee} 
	\ergo(\dstate;\ham) =  \en(\dstate;\ham)-  \en(\dstate_{\text{pass}};\ham)  \;,
	\end{equation}
with $\dstate_{\text{pass}}$
obtained  by mapping the  eigenvectors $\{ |\lambda_\ell\rangle\}_{\ell}$ of $\dstate$ into the eingenvectors $\{ |E_\ell\rangle\}_{\ell}$ 
of the system Hamiltonian $\ham$, matching the corresponding eingenvalues in reverse order, i.e. 
 \begin{eqnarray} 
 \left.
 \begin{array}{l} 
 \dstate =\sum_{\ell =1}^{d} \lambda_\ell |\lambda_\ell\rangle \! \langle\lambda_\ell|  \\ \\
\ham=\sum_{\ell =1}^{d} E_\ell |E_\ell\rangle\!\langle E_\ell| 
 \end{array} \right\} \mapsto
 \dstate_{\text{pass}} := \sum_{\ell =1}^{d} \lambda_\ell |E_\ell\rangle\!\langle E_\ell| \;, \nonumber \\ \label{deftutto} 
 \end{eqnarray} 
 where for all $\ell=1,\cdots, d-1$,  we set  $\lambda_{\ell} \geq \lambda_{\ell+1}$ and 
 $E_\ell  \leq E_{\ell+1}$.

The
	total-ergotropy $\ergo_{\rm tot}(\dstate;\ham)$ 
 is a regularized version of  Eq.~(\ref{ergoeeeDEFinvariance})
 that emerges when considering  scenarios where 
one has at disposal an arbitrary large number of identical copies of the input state $\dstate$. 
Formally it is defined as 
\begin{eqnarray}\label{totergoeeeDEFtotal} 
\ergo_{\rm tot}(\dstate;\ham)  &:=& \lim_{N\rightarrow \infty} \frac{\ergo(\dstate^{\otimes N};\ham^{(N)})}{N}\;,
	\end{eqnarray}
	where for fixed $N$ integer, $\ham^{(N)}$ is the total Hamiltonian of the $N$ copies of the system  
obtained by assigning to each of them the same $\ham$ (no interactions being included). 
One can verify~\cite{Alicki2013,def_ergo} that  the limit in Eq.~(\ref{totergoeeeDEFtotal}) exists and 
 corresponds to the maximum of $\frac{\ergo(\dstate^{\otimes N};\ham^{(N)})}{N}$ with respect to all possible $N$, implying in particular that $\ergo_{\rm tot}(\dstate;\ham)$ is at least as large as $\ergo(\dstate;\ham)$,
 i.e.
 \begin{eqnarray} \label{qubitimpo1} 
 \ergo_{\rm tot}(\dstate;\ham) =\sup_{N\geq 1} \frac{\ergo(\dstate^{\otimes N};\ham^{(N)})}{N}  \geq \ergo(\dstate;\ham)\;.
 \end{eqnarray} 
The case where the system has dimension 2 represents an exception to this rule since under this condition one has
  \begin{eqnarray} \label{qubitsimpo} 
 \ergo_{\rm tot}(\dstate;\ham) =\lim_{N\to\infty}\frac{\ergo(\dstate^{\otimes N};\ham^{(N)})}{N}  =\ergo(\dstate;\ham)\;, 
 \end{eqnarray} 
 for all inputs and for all $N$. 
Notably $\ergo_{\rm tot}(\dstate;\ham)$ can be  expressed as the energy gap between the 
state $\dstate$ and its completely passive counterpart $ \dstate_{\text{c-pass}}$, i.e. 
\begin{equation} \label{GIBBS} 
\ergo_{\rm tot}(\dstate;\ham) =
 \en(\dstate;\ham)-  \en( \dstate_{\text{c-pass}}; \ham)\;,
\end{equation} 
with $\dstate_{\text{c-pass}}$ being the thermal Gibbs state 
\begin{eqnarray} \label{GIBBS1} 
\hat{\tau}_{\beta_{\star}}:= e^{ - \beta_{\star} \ham}/Z({\beta_{\star}})\;, \qquad  Z({\beta_{\star}}) : =\mbox{Tr}[ e^{ - \beta_{\star} \ham}]  \;, \end{eqnarray} 
with effective inverse temperature $\beta_\star$  selected to ensure that 
$\dstate_{\text{c-pass}}$ has the same von Neumann entropy of $\dstate$, i.e. 
\begin{eqnarray} S(\dstate_{\text{c-pass}})= S(\dstate) := -\mbox{Tr}[\dstate \ln \dstate]\;.\end{eqnarray}

 Clearly both  $\dstate_{\text{pass}}$  and $\dstate_{\text{c-pass}}$ are complicated  functions of $\dstate$:
 notice  however that they only depend upon the spectrum of such operator, meaning that input states differing by a unitary rotation will have the same passive and completely passive counterparts.

\section{Explicit definitions of the quantum work capacitances} \label{app:newdef}
Here we review the definitions of the ergotropic capcitances that can be found in \cite{quantumworkcapacitances}. We start by the definition of the local ergotropic capacitance $C_{\rm loc}(\Lambda;\mathfrak{e})$ for any quantum channel $\Lambda$ and any $\mathfrak{e}\in[0,||\hat{h}||_{\infty}]$. At first we define the $n$-uses maximal output local ergotropy:

\begin{eqnarray}
&& \ergo_{\rm loc}^{(n)}(\Lambda;E) :=  \max_{\dstate^{(n)}\in\mathfrak{S}_{E}^{(n)}}\bigg[\en(\Lambda^{\otimes n}(\dstate^{(n)}));\ham^{(n)}) - \nonumber \\
&& \min_{\hat{U}_1,..,\hat{U}_n\in\mathbb{U}(d)}\en(\hat{U}_1\otimes...\otimes\hat{U}_n\Lambda^{\otimes n}(\dstate^{(n)})\hat{U}_1^{\dagger}\otimes...\otimes\hat{U}_n^{\dagger};\ham^{(n)}) \bigg] \; , \nonumber \\
\end{eqnarray}
so we can now define the local ergotropic capacitance as
\begin{equation}
C_{\rm loc}(\Lambda,\mathfrak{e}) := \lim_{n\to\infty}\frac{\ergo_{\rm loc}^{(n)}(\Lambda,n\mathfrak{e})}{n} \; .
\end{equation}
The separable input $n$-uses maximal output ergotropy for any channel $\Lambda$ and any input energy $\mathfrak{e}$ is:
\begin{eqnarray}
\ergo_{\rm sep}^{(n)}(\Lambda;E) := \max_{\dstate_{\rm sep}\in\mathfrak{S}_{E,\rm sep}^{(n)}}\bigg[\en(\Lambda^{\otimes n}(\dstate^{(n)}_{\rm sep});\ham^{(n)}) - \nonumber \\ 
 \min_{\hat{U}\in\mathbb{U}(d^n)}(\en(\hat{U}\Lambda^{\otimes n}(\dstate^{(n)}_{\rm sep})\hat{U}^{\dagger});\ham^{(n)})\bigg] \; ,
\end{eqnarray}
here $\mathfrak{S}^{(n)}_{E,\rm sep}$ represents the set of all separable $n$ qudit states with input energy $\en(\dstate^{(n)})_{\rm sep} \leq E$. So we can now define the separable input ergotropic capacitance as
\begin{equation}
C_{\rm sep}(\Lambda;\mathfrak{e}) := \lim_{n\to\infty}\frac{\ergo_{\rm sep}^{(n)}(\Lambda;n\mathfrak{e})}{n} \; .
\end{equation}
Finally in order to define the $C_{\rm loc,sep}(\Lambda,\mathfrak{e})$ capacitance we write the local separable maximal output ergotropy as 
\begin{eqnarray}
&& \ergo_{\rm loc,sep}^{(n)}(\Lambda,E) := \max_{\dstate^{(n)}_{\rm sep}\in\mathfrak{S}^{(n)}_{E,\rm sep}}\bigg[\en(\Lambda^{\otimes n}(\dstate^{(n)}_{\rm sep});\ham^{(n)}) - \nonumber \\
&& \min_{\hat{U}_1,..,\hat{U}_n\in\mathbb{U}(d)}\en(\hat{U}_1\otimes...\otimes\hat{U}_n\Lambda^{\otimes n}(\dstate^{(n)}_{\rm sep})\hat{U}_1^{\dagger}\otimes...\otimes\hat{U}_n^{\dagger};\ham^{(n)}) \bigg] \; , \nonumber \\
&&
\end{eqnarray}
the the capacitance for any channel $\Lambda$ and input energy $\mathfrak{e}$ is 
\begin{equation}
C_{\rm loc,sep}(\Lambda;\mathfrak{e}) := \lim_{n\to\infty}\frac{\ergo^{(n)}_{\rm loc,sep}(\Lambda;n\mathfrak{e})}{n} \; .
\end{equation}
The existence of the limits of the functions above is guaranteed by Fekete's lemma as proven in \cite{quantumworkcapacitances}.

\section{Proof of Corollary~\ref{cor1} }\label{AppeA} 
Here we give an explicit proof of Corollary~\ref{cor1}, which shows how  to lower bound the 
separable-input capacitance $C_{\rm sep} \left( \Lambda; \mathfrak{e} \right)$ with the convex closure 
$\chi_{\rm tot}(\Lambda; \mathfrak{e})$ of the single-shot energy constrained maximum total ergotropy
$\ergo_{\rm tot}^{(1)}(\Lambda;\mathfrak{e})$~\cite{quantumworkcapacitances}, i.e. the quantity 
 \begin{eqnarray} 
 \ergo_{\rm tot}^{(1)}(\Lambda;\mathfrak{e})&:=& \max_{\dstate \in \mathfrak{S}_{\mathfrak{e}}}
 \ergo_{\rm tot}(\Lambda(\dstate);\hat{h}) \;,   \label{dergotot}
 \end{eqnarray}
 where the maximization is performed over the set $\mathfrak{S}_{\mathfrak{e}}$ of the single q-cell states
 with input energy not larger than~${\mathfrak{e}}$.

\begin{proof}
For large enough $n$ consider the separable density matrix 
$\dstate^{(n)} := \dstate_1^{\otimes m} \otimes 
\dstate_2^{\otimes m} \otimes \cdots \otimes \dstate_k^{\otimes m}$,  
where $m$ and $k$ are integers that give a partition of $n$, i.e.
$n= m k$,
and $\dstate_1$, $\dstate_2$, $\cdots$, $\dstate_k$ are single site density matrices with mean energies
$\mathfrak{e}_1$, $\cdots$, $\mathfrak{e}_k$, chosen in such a way that 
$\dstate^{(n)}$ has mean energy  smaller than or equal to $E= n \mathfrak{e}$, i.e.
 \begin{eqnarray} \label{cons11} 
 \sum_{i=1}^k m {\mathfrak{e}}_i \leq n \mathfrak{e}  \qquad \Longleftrightarrow \qquad 
  \sum_{i=1}^k \frac{1}{k} {\mathfrak{e}}_i \leq  \mathfrak{e} \;.
 \end{eqnarray} 
We can hence write
 \begin{eqnarray} \label{ecco} 
 \frac{1}{n} \ergo_{\rm sep}^{(n)}(\Lambda;E=n\mathfrak{e})&\geq& 
 \frac{1}{n} \ergo (\Lambda^{\otimes n}( {\dstate}^{(n)});\ham^{(n)}) \nonumber \\
 &\geq &\frac{1}{n} \sum_{i=1}^k  \ergo (\Lambda^{\otimes m}( {\dstate}_i^{\otimes m});\ham^{(m)}) 
 \nonumber \\
 &\geq &\frac{1}{m k} \sum_{i=1}^k  \ergo \left(\left(\Lambda( {\dstate}_i)\right)^{\otimes m};\ham^{(m)}\right) \nonumber\;,
 \end{eqnarray} 
 where the second inequality is obtained by restricting the optimization over the set of unitaries acting locally on the $k$ blocks of q-cells. Taking $m\rightarrow \infty$ the first term of~(\ref{ecco})
 gives $C_{\rm sep} \left( \Lambda; \mathfrak{e} \right)$. Vice versa, since for all $i\in\{ 1,\cdots, k\}$ we have 
\begin{eqnarray} \lim_{m \rightarrow \infty}  \frac{\ergo\left(\left(\Lambda( {\dstate}_i)\right)^{\otimes m};\ham^{(m)}\right)}{m}=
 \ergo_{\rm tot} (\Lambda( {\dstate}_i);\hat{h})\;, \end{eqnarray}  
 we get 
\begin{eqnarray} C_{\rm sep} \left( \Lambda; \mathfrak{e} \right)\geq  \sum_{i=1}^k \frac{1}{k}   \ergo_{\rm tot} (\Lambda( {\dstate}_i);\hat{h})\;.
\end{eqnarray}  
Next step is to apply the above inequality to a collection of states 
$\dstate_1$, $\dstate_2$, $\cdots$, $\dstate_k$ which saturate the single-shot maximal output total ergotropy 
of the model for the energies ${\mathfrak{e}}_1$, ${\mathfrak{e}}_2$, $\cdots$, ${\mathfrak{e}}_n$.
Accordingly, we can replace all the $\ergo_{\rm tot} (\Lambda( {\dstate}_i);\hat{h})$ with the associated 
$\ergo^{(1)}_{\rm tot} (\Lambda;{\mathfrak{e}}_i)$, leading to 
 \begin{eqnarray} \label{lowerforcsep11} 
C_{\rm sep} \left( \Lambda; \mathfrak{e} \right)\geq  \sum_{i=1}^k \frac{1}{k}   \ergo^{(1)}_{\rm tot} (\Lambda;{\mathfrak{e}}_i)\;.
\end{eqnarray} 
The thesis finally follows by observing that 
since we can choose both $k$ and the ${\mathfrak{e}}_i$ arbitrarily within the constraint~(\ref{cons11}), the r.h.s. of
(\ref{lowerforcsep11}) can reproduce all the averages $\sum_j p_j\;  \ergo_{\rm tot}^{(1)}(\Lambda;\mathfrak{e}_j)$ fulfilling (\ref{exact}) of the main text. 
\end{proof}

\section{Harmonic oscillator models}
\label{app:oscillatorearmonico}
Consider a collection of  $k$ independent harmonic oscillators characterized by the same frequency $\omega$, representing e.g. $k$ spatial modes of the
electromagnetic field~\cite{SERAFINIBOOK}. 
We can formally describe it via  a complex separable infinite-dimensional Hilbert space with self-adjoint operators $\pos_1,\mom_1,...,\pos_k,\mom_k$ satisfying the canonical commutation relations: $[\pos_j,\pos_k]=0$ $[\mom_j,\mom_k]=0$ $[\pos_j,\mom_k]=\delta_{j,k}$. Its Hamiltonian writes
\begin{equation}
\ham_k = \sum_{j=1}^k \hbar\omega \;  \hat{a}_j^{\dagger}\hat{a}_j,
\end{equation}
where we subtracted the term $\hbar\omega/2$ in order to consider a ground state having zero energy and 
where we introduced the annihilation and creation operators $\hat{a}_j := \sqrt{\frac{\omega}{2\hbar}}\left(\pos_j + \frac{i}{\omega}\mom_j \right)$, $\hat{a}^{\dagger} := \sqrt{\frac{\omega}{2\hbar}}\left(\pos_j - \frac{i}{\omega}\mom_j \right)$. 
Given $\dstate$ a generic state of the model,  we define its characteristic function as $\xi(\dstate;x) := \Tr[\dstate \hat{D}(x)]$, where $\hat{D}(x): = e^{i\hat{r}\cdot x}$ is the displacement operator, $x\in\mathbb{R}^{2k}$ and $\hat{r} := (\pos_1,\mom_1,...,\pos_k,\mom_k)^T$. 
An important  class of  noise models for these systems  is provided by the Bosonic Gaussian Channels (BGCs)~\cite{HOLEVOBOOK, SERAFINIBOOK} which  characterize dissipation, thermalization, amplification, and squeezing effects. 
Using the characteristic function formalism, the  action of a $k$-mode BGC channel $\Phi_k$ is  assigned via the mapping 
 \begin{eqnarray} \xi(\Phi_k(\dstate);x) = \xi(\dstate; X^Tx)e^{-\frac{1}{4}x^TYx+iv\cdot x}\;,\end{eqnarray} 
 where $v$ is an assigned vector of  $\mathbb{R}^{2k}$ and  where 
 $X$ and $Y$ 
are two  real $2k\times 2k$  matrices 
satisfying the conditions 
\begin{eqnarray} Y \geq i(\theta_k - X\theta_k X^T)\;, \end{eqnarray}  with $\theta_k := \begin{pmatrix} 0 && I_k \\ -I_k && 0 \end{pmatrix}$ is the symplectic form of the model (in this expression $I_k$ stands for $k\times k$ identity matrix). 
%
%
A special  subset of BGCs is represented by the so called phase-insensitive bosonic Gaussian Channels (PI-BGCs), which satisfies the following additional condition
\begin{equation}
\Phi_k(e^{-i\ham_k t}\dstate e^{i\ham_kt}) = e^{\mp i\ham_kt}\Phi_k({\dstate})e^{\pm i\ham_kt},
\end{equation}
for any quantum state $\dstate$ and $t\in\mathbb{R}$.  Single mode ($k=1$)
PI-BGCs include 
\begin{itemize}
\item thermal attenuators $\mathcal{L}_{\lambda,N}$ ($v=0$, $X= \sqrt{\lambda} I_2$, $Y= (1-\lambda)(2N+1)I_2$, with $\lambda \in [0,1)$ and $N\geq 0$), 
\item 
thermal amplifiers $\mathcal{A}_{\mu,N}$ ($v=0$, $X = \sqrt{\mu}I_2$, $Y = (\mu -1)(2N+1)I_2$,  with $\mu > 1$ and $N\geq 0$), 
\item additive noise channels 
 $\mathcal{N}_N$ ($v=0$,$X= I_2$ and $Y = 2NI_2$, with $N\geq 0$).
 \end{itemize} 
The maximum values for the 
ergotropy and total-ergrotopy functionals at the output of PI-BGCs were studied in Ref.~\cite{PhysRevLett.127.210601} showing 
 that optimal inputs are provided by the multi-mode coherent states 
  \begin{eqnarray}
 \ket{\vec{\varphi}}&:=&|\varphi_{1}\rangle\otimes \cdots \otimes |\varphi_k\rangle\;, 
 \end{eqnarray} 
where for $j=1,\cdots, k$,  $|\varphi_j\rangle$ is a coherent state of the $j$-th mode with complex amplitude $\varphi_j$.
 Specifically, given 
 $\Phi_k$ a generic PI-BGC it follows that, 
\begin{eqnarray}\label{resulta} 
\ergo\left(\Phi_k(\ket{\vec{\varphi}}\!\!\bra{\vec{\varphi}});\ham_k\right)
&\geq&
\ergo\left(\Phi_k(\dstate);\ham_k \right) \; ,\\
\ergo_{\rm tot} \left(\Phi_k(\ket{\vec{\varphi}}\!\!\bra{\vec{\varphi}});\ham_k\right)
&\geq&
\ergo_{\rm tot} \left(\Phi_k(\dstate);\ham_k \right) \; , \nonumber 
\end{eqnarray}
for all 
 states $\dstate$ with 
 \begin{eqnarray} \en(\dstate; \ham_k) \leq  \en(\ket{\vec{\varphi}}\!\!\bra{\vec{\varphi}};\ham_k) =
 \hbar \omega |\vec{\varphi}|^2 =  \hbar \omega \sum_{j=1}^k |\varphi_j|^2\;.
 \end{eqnarray} 

\subsection{Capacitance gap for $k=2$-mode BGCs}
Here we evaluate the ergotropic capacitances of a PI-BGC  
$\Phi_2 = \mathcal{L}_{\lambda,N}\otimes\mathcal{L}_{0,0}$ acting on two independent, resonant modes 
(mode $A$ and mode $A'$  both having frequency $\omega$). 
From~(\ref{resulta}) it follows that computing $\ergo^{(n)}(\Phi_2;\mathfrak{e})$ and 
$\ergo_{\rm tot}^{(n)}(\Phi_2;\mathfrak{e})$ with $n$ arbritrary integer we can restrict the analysis of the
ergotropic capacitances of $2n$-mode coherent input states 
\begin{eqnarray}
|\vec{\varphi}\rangle:=|\varphi_1\rangle\otimes |\varphi_1' \rangle\otimes \cdots
\otimes |\varphi_n\rangle\otimes |\varphi_n' \rangle\;, 
\end{eqnarray}
 fulfilling the energy constraint
\begin{eqnarray}  \en(\ket{\vec{\varphi}}\!\!\bra{\vec{\varphi}};\ham_2^{(n)}) &=& \hbar \omega |\vec{\varphi}|^2\nonumber \\
&=& \hbar \omega\sum_{j=1}^n \left( |{\varphi}_j|^2+|{\varphi}_j'|^2\right)\leq n \mathfrak{e}  \;,\label{eed} 
\end{eqnarray} 
(in the above expressions $|\varphi_j\rangle$ and $|\varphi_j'\rangle$ represent
coherent states of $j$-copies of the modes $A$ and $A'$ respectively, while 
 $\hat{H}_2$ is their Hamiltonian).
Observe that the single-mode channel $\mathcal{L}_{0,0}$ maps all input states of the mode $A'$ into the ground state configuration $|0\rangle$, while
$\mathcal{L}_{\lambda,N}$ maps coherent states of mode $A$ into thermal displaced density matrices charaterized by inverse 
temperature 
\begin{eqnarray}\label{defbeta} 
\beta_{\lambda,N}:= \log\frac{(1-\lambda)N+1}{(1-\lambda)N}\;. 
\end{eqnarray} 
Specifically indicating with $\hat{D}_1$ the displacement operator that applied to the ground of mode $A$ 
generates the coherent state 
 $|\sqrt{\lambda} \varphi_1\rangle$, 
we can write 
\begin{equation}
\Phi_2(\ket{{\varphi}_1}\!\!\bra{{\varphi}_1}\otimes \ket{{\varphi}'_1}\!\!\bra{{\varphi}'_1}) = \hat{D}_1 \left( \hat{\tau}_{\beta_{\lambda, N}}\otimes \ket{0}\!\!\bra{0} \right)
\hat{D}_1^{\dagger}\;, 
\label{output_coerenti}
\end{equation}
where $\hat{\tau}_{\beta_{\lambda, N}}=\frac{e^{-\beta_{\lambda, N} \hat{h}}}{Z_1\left(\beta_{\lambda, N}\right) }$, 
with 
$\hat{h} := \hbar \omega \hat{a}^\dag \hat{a}$ the single mode Hamiltonian, and 
 $Z_1\left(\beta\right): = \mbox{Tr}[ e^{-\beta \hat{h}}]$ the associated single mode partition function. 
 Generalizing to the case of $n$ copies this leads to 
 \begin{eqnarray}
\Phi^{\otimes n}_2(\ket{\vec{\varphi}}\!\!\bra{\vec{\varphi}}) &=& 
\bigotimes_{j=1}^n  \Phi_2(\ket{{\varphi}_j}\!\!\bra{{\varphi}_j}\otimes \ket{{\varphi}'_j}\!\!\bra{{\varphi}'_j}) 
\nonumber \\
&=& \hat{D}_A \left( \hat{\tau}_{\beta_{\lambda, N}}\otimes \ket{0}\!\!\bra{0} \right)^{\otimes n} 
\hat{D}^\dag_A
\label{output_coerentin}
\end{eqnarray}
with $\hat{D}_A = \hat{D}_1\otimes \hat{D}_2\otimes \cdots\otimes \hat{D}_n$. 
 From 
from Eq.~(\ref{ergoeee}) and (\ref{GIBBS}) we get 
\begin{eqnarray} \label{GIBBScoer} 
&&\ergo(\Phi^{\otimes n}_2(\ket{\vec{\varphi}}\!\!\bra{\vec{\varphi}});\ham^{(n)}_2) \\ 
&&\qquad = \nonumber 
 \en(\Phi^{\otimes n}_2(\ket{\vec{\varphi}}\!\!\bra{\vec{\varphi}});\ham^{(n)}_2)-  \en(\dstate^{(n)}_{\text{pass}};\ham^{(n)}_2)\;,  \end{eqnarray}
and 
\begin{eqnarray} 
&&\ergo_{\rm tot}(\Phi^{\otimes n}_2(\ket{\vec{\varphi}}\!\!\bra{\vec{\varphi}});\ham^{(n)}_2) \\ 
&&\qquad = \nonumber 
 \en(\Phi^{\otimes n}_2(\ket{\vec{\varphi}}\!\!\bra{\vec{\varphi}});\ham^{(n)}_2)-  \en(\dstate^{(n)}_{\text{c-pass}};\ham^{(n)}_2)\;,  
%
\end{eqnarray} 
with the  output mean energy that can be easily computed as
\begin{equation} \en(\Phi^{\otimes n}_2(\ket{\vec{\varphi}}\!\!\bra{\vec{\varphi}});\ham^{(n)}_2)
= \hbar \omega \left(   \sum_{j=1}^n\lambda |{\varphi}_j|^2 + (1-\lambda) N n\right) \;. 
 \end{equation} 
In the above expressions $\dstate^{(n)}_{\text{pass}}$ and $\dstate^{(n)}_{\text{c-pass}}$
are respectively
the passive and completely passive counterparts of $\Phi^{\otimes n}_2(\ket{\vec{\varphi}}\!\!\bra{\vec{\varphi}})$.
 Exploiting the
unitary invariance discussed at the end of first section of the Supp. Mat. and the identity~(\ref{output_coerentin}) it also
follows that $\dstate^{(n)}_{\text{pass}}$ and $\dstate^{(n)}_{\text{c-pass}}$
correspond to the passive and completely passive counterparts of 
the $2n$-mode product state
\begin{eqnarray} 
\dstate^{(n)}_{\lambda, N} : = \left( \hat{\tau}_{\beta_{\lambda, N}}\otimes \ket{0}\!\!\bra{0} \right)^{\otimes n} \;,
\end{eqnarray} 
which does not bear any dependence upon the parameters $\varphi_{j}$'s and 
$\varphi'_{j}$'s and thus on the  energy at the input of the channel. 
Accordingly we can conclude that neither
$\en(\dstate^{(n)}_{\text{pass}};\ham^{(n)}_2)$ nor $\en(\dstate^{(n)}_{\text{c-pass}};\ham^{(n)}_2)$ are functions of~$\mathfrak{e}$
allowing us to solve explicitly 
the maximization required to compute  $\ergo^{(n)}(\Phi_2;\mathfrak{e})$ and 
$\ergo_{\rm tot}^{(n)}(\Phi_2;\mathfrak{e})$: 
\begin{eqnarray} \label{ergon}
\frac{\ergo^{(n)}(\Phi_2;n\mathfrak{e})}{n\hbar \omega} &=&\frac{\lambda\mathfrak{e}}{\hbar\omega}+ (1-\lambda) N -
\frac{ \en(\dstate^{(n)}_{\text{pass}};\ham^{(n)}_2)}{n\hbar\omega},\\ \label{ergon1} 
\frac{\ergo_{\rm tot}^{(n)}(\Phi_2;n\mathfrak{e})}{n\hbar \omega}&=&
   \frac{\lambda\mathfrak{e}}{\hbar\omega}+ (1-\lambda) N -\frac{ \en(\dstate^{(n)}_{\text{c-pass}};\ham^{(n)}_2)}{n\hbar\omega} \; . \nonumber \\   
\end{eqnarray} 
Following the prescription of  Theorem~\ref{th1} to compute the value of  
$C_{\rm loc} \left( \Phi_2; \mathfrak{e} \right)$ (and $C_{\rm loc, sep} \left( \Phi_2; \mathfrak{e} \right)$) we need just to consider the convex clousure of function~(\ref{ergon}) evaluated for
 $n=1$, i.e.
 \begin{eqnarray} \label{exactnewnew} 
\frac{C_{\rm loc} \left( \Phi_2; \mathfrak{e} \right)}{\hbar \omega} &=&\frac{\lambda\mathfrak{e}}{\hbar\omega}+ (1-\lambda) N -
\frac{\en(\dstate^{(1)}_{\text{pass}};\ham_2)}{\hbar \omega}\;. 
\end{eqnarray} 
where now  $\dstate^{(1)}_{\text{pass}}$ the passive counterpart of the two-modes state
\begin{eqnarray} \dstate^{(1)}_{\lambda, N} := \hat{\tau}_{\beta_{\lambda, N}}\otimes \ket{0}\!\!\bra{0}\;.\label{target11} 
\end{eqnarray} 

Regarding $C_{\rm sep} \left( \Phi_2; \mathfrak{e} \right)$ we need instead to take the limit over $n\rightarrow \infty$ of
(\ref{ergon}), or equivalently, of (\ref{ergon1})~\cite{quantumworkcapacitances}.
The calculation simplifies by observing  that  the completely passive counterpart of the tensor product state
$\dstate^{(n)}_{\lambda, N}$ corresponds to 
\begin{eqnarray}
\dstate^{(n)}_{\text{c-pass}}= \dstate_{\text{c-pass}}^{\otimes n}\;, 
\end{eqnarray} 
where $\dstate_{\text{c-pass}}$ is the completely passive of $\dstate^{(1)}_{\lambda, N}$. Thus we can write 
 \begin{eqnarray} \label{exact1} 
\frac{C_{\rm sep} \left( \Phi_2; \mathfrak{e} \right)}{\hbar \omega} &=&\frac{\lambda\mathfrak{e}}{\hbar\omega}+ (1-\lambda) N -
\frac{\en(\dstate_{\text{c-pass}};\ham_2)}{\hbar \omega}\;,
\end{eqnarray} 
(notice that, thanks to (\ref{resulta}) 
the above expression also provides the exact value of $C_{\cal E} \left( \Phi_2; \mathfrak{e} \right)$).
From Eqs.~(\ref{exact}) and (\ref{exact1}) we finally obtain \begin{eqnarray} \nonumber 
\Delta C(\mathfrak{e})  &=&C_{\rm sep} \left( \Phi_2; \mathfrak{e} \right)-C_{\rm loc} \left( \Phi_2; \mathfrak{e} \right)\\
&=&\en(\dstate_{\text{pass}};\ham_2)-\en(\dstate_{\text{c-pass}};\ham_2)\;,\label{ddfddg} 
\end{eqnarray} 
which does not depend upon the input energy $\mathfrak{e}$. 
To show that the  gap~(\ref{ddfddg}) is non zero, it is sufficient to verify that 
\begin{eqnarray} \dstate^{(1)}_{\text{pass}}\neq \dstate^{(1)}_{\text{c-pass}}\label{proof1} \;, \end{eqnarray} i.e.
that the passive counterpart of (\ref{target11}) is not completely passive. 
To see this, let us invoke Eq.~(\ref{GIBBS}) to express
\begin{eqnarray} \label{cpassthermal} 
\dstate^{(1)}_{\text{c-pass}} = \frac{e^{ - \beta_{\star} \ham_2}}{\mbox{Tr}[e^{ - \beta_{\star} \ham_2}]}=\left( 
\frac{e^{ - \beta_{\star} \hat{h}}}{Z_1({\beta_{\star}})}\right)^{\otimes 2} = 
\hat{\tau}_{\beta_{\star}}^{\otimes 2}\;,
\end{eqnarray} 
where $\beta_{\star}$ is chosen so that 
\begin{eqnarray}S(\hat{\tau}_{\beta_{\star}}^{\otimes 2}) = 2 S(\hat{\tau}_{\beta_{\star}})&=& 
 S(\dstate^{(1)}_{\lambda, N}) = 
S(\hat{\tau}_{\beta_{\lambda, N}})  \nonumber \\
  &\Longrightarrow& \beta_{\star}= {\cal S}^{-1}\left({\cal S}(\beta_{\lambda, N}) /2\right)\;, 
 \label{ddfd222} 
\end{eqnarray} 
with ${\cal S}(x)$ the function \begin{equation}
S(\hat{\tau}_{\beta}) = {\cal S}(\beta) := -\ln\left( 1- e^{-\beta \hbar \omega} \right)
+\beta\hbar \omega \left(\frac{  e^{-\beta\hbar \omega}}{1- e^{-\beta\hbar \omega}}\right)\;, 
\end{equation}
which expresses  the von Neumann entropy of the single-mode thermal states $\hat{\tau}_{\beta}$.
Equation~(\ref{cpassthermal}) implies that all the eigenvalues of $\dstate_{\text{c-pass}}$
have explicit degeneracies, 
\begin{eqnarray} \begin{array}{lcc} 
\lambda_0(\dstate^{(1)}_{\text{c-pass}}) = \tfrac{1}{Z^2_1({\beta_{\star}})}\;,&\quad & \mbox{deg $=1$,} \\ \\
\lambda_1(\dstate^{(1)}_{\text{c-pass}}) =\tfrac{e^{-\beta_{\star} \hbar \omega}}{Z^2_1({\beta_{\star}})}
\;, & \quad & \mbox{deg $=2$,} \\
\\
\lambda_2(\dstate^{(1)}_{\text{c-pass}}) =\tfrac{e^{-2\beta_{\star} \hbar \omega}}{Z^2_1({\beta_{\star}})}
\;, & \quad & \mbox{deg $=3$,} \\ \\
\cdots & \quad & \cdots \\ \\ 
\lambda_k(\dstate^{(1)}_{\text{c-pass}}) =\tfrac{e^{-k\beta_{\star} \hbar \omega}}{Z^2_1({\beta_{\star}})}
\;, & \quad & \mbox{deg $=k+1$,} 
\end{array} \label{signature} 
\end{eqnarray} 
If by contradiction $\en(\dstate^{(1)}_{\text{c-pass}};\ham_2)$ would coincide with $\en(\dstate^{(1)}_{\text{pass}};\ham_2)$, then for any energy eigenvalue $k$ the sum of its $k+1$ eigenvalues of $\dstate^{(1)}_{\text{pass}}$ must be equal to $(k+1)e^{-\beta_{\star}\hbar\omega}/Z_1^2(\beta_{\star})$ due to the degeneracy of the eigenspace of $\ham_2$. This fact sets an hard constraint on the positive eingenvalues of $\dstate^{(1)}_{\text{pass}}$: 
\begin{eqnarray}\label{spectre} 
\lambda_k(\dstate^{(1)}_{\text{pass}}) = \frac{e^{-k\beta_{\lambda, N} \hbar \omega}}{Z_1({\beta_{\lambda, N}})}\;,
\end{eqnarray}
in particular their ratios must satisfy the following conditions
\begin{equation}
e^{\hbar\omega\beta_{\star}} = \frac{2\lambda_0(\dstate^{(1)}_{\text{pass}})}{\lambda_1(\dstate^{(1)}_{\text{pass}}) + \lambda_2(\dstate^{(1)}_{\rm pass})} = \frac{2}{e^{-\hbar\omega\beta_{\lambda,N}} + e^{-2\hbar\omega\beta_{\lambda,N}}} \; ,
\end{equation}
and 
\begin{eqnarray}
e^{\hbar\omega\beta_{\star}} &=& \frac{3}{2}\frac{\lambda_1(\dstate^{(1)}_{\rm pass}) + \lambda_2(\dstate^{(1)}_{\rm pass})}{\lambda_3(\dstate^{(1)}_{\rm pass}) + \lambda_4(\dstate^{(1)}_{\rm pass}) + \lambda _5(\dstate^{(1)}_{\rm pass})} \nonumber \\
&=& \frac{3}{2}\frac{e^{-\hbar\omega\beta_{\lambda,N}}+e^{-2\hbar\omega\beta_{\lambda,N}}}{e^{-3\hbar\omega\beta}+e^{-4\hbar\omega\beta}+e^{-5\hbar\omega\beta}} \; ,
\end{eqnarray}
so that would imply
\begin{equation}
 \tfrac{2}{e^{-\hbar\omega\beta_{\lambda,N}} + e^{-2\hbar\omega\beta_{\lambda,N}}} = \tfrac{3}{2}\tfrac{e^{-\hbar\omega\beta_{\lambda,N}}+e^{-2\hbar\omega\beta_{\lambda,N}}}{e^{-3\hbar\omega\beta_{\lambda,N}}+e^{-4\hbar\omega\beta_{\lambda,N}}+e^{-5\hbar\omega\beta_{\lambda,N}}} \; ,
\end{equation}
but for sufficiently large $\beta_{\lambda,N}$ this last condition cannot be satisfied. \\
We now explicitly compute the values of $\en(\dstate^{(1)}_{\text{pass}};\ham_2)$ and $\en(\dstate^{(1)}_{\text{c-pass}};\ham_2)$.
On one hand, from~(\ref{cpassthermal}) we know that the latter is nothing but the mean energy of the two-mode thermal state $\hat{\tau}_{\beta_{\star}}^{\otimes 2}$, i.e. 
\begin{equation} \label{onone} 
\en(\dstate^{(1)}_{\text{c-pass}};\ham_2) = 2 \en(\hat{\tau}_{\beta_{\star}};\hat{h} ) = 
 2 \hbar \omega \left(\frac{e^{-\beta_{\star}\hbar \omega}}{1- e^{-\beta_{\star}\hbar \omega}}\right) \;,
\end{equation} 
which an implicit function of $\beta_{\lambda, N}$ via~(\ref{ddfd222}). 
On the other hand, from~(\ref{spectre}) we get 
\begin{eqnarray} \label{onother} 
&&\en(\dstate^{(1)}_{\text{pass}};\ham_2) 
=   \hbar \omega \sum_{k=0}^{\infty} k \sum_{k'\in I_k} \frac{e^{-k\beta_{\lambda, N} \hbar \omega}}{Z_1({\beta_{\lambda, N}})}\\
&&=  \hbar \omega \sum_{k=0}^{\infty}k \left[e^{-\frac{k(k+1)}{2}\hbar\omega\beta_{\lambda, N}} - e^{-(\frac{k(k+3)}{2}+1)\hbar\omega\beta_{\lambda, N}} \right]
\;, \nonumber 
\end{eqnarray} 
where  $I_k$ represents the collection of $k+1$ integers defined by the following sequence:
\begin{eqnarray} I_0 &=&\{ 0\} \;, \qquad  I_1 =\{ 1,2\}\;, \qquad I_2 =\{ 3,4,5\}\;,  \nonumber \\
 I_3 &=&\{ 6,7,8,9\}\;, \qquad \cdots
\end{eqnarray} 
Panel b) of Fig. \ref{fig:MADworld} was realized by numerically evaluating~(\ref{onone}) and (\ref{onother}).

\subsection{Multilevel Amplitude Damping channels}

Multilevel Amplitude Damping (MAD) channels~\cite{MAD}, describe relaxation processes of excited to less energetic states in qudit systems. 
In the case of a qutrit, the possible decay processes are associated with three damping parameters $\{ \gamma_1, \gamma_2,  \gamma_3 \}$ that define  $\Phi_{\gamma_1, \gamma_2, \gamma_3}$  via the following transformation \begin{widetext}
\begin{equation}\label{eq: channel actionMAD}
\hspace*{-0.3cm}
\Phi_{\gamma_1, \gamma_2, \gamma_3}(\hat{\rho})=\left(
\begin{array}{ccc}
\rho_{00} + \gamma_{1} \rho_{11}+\gamma_{3} \rho_{22} &  \sqrt{1-\gamma_{1}} \rho_{01}  &
    \sqrt{1-\gamma_{2}-\gamma_{3}} \rho_{02} \\
\sqrt{1-\gamma_{1}} \rho_{01}^*  & (1-\gamma_{1}) \rho_{11}+\gamma_{2} \rho_{22} &
    \sqrt{(1-\gamma_{1}) (1-\gamma_{2}-\gamma_{3})} \rho_{12} \\
 \sqrt{1-\gamma_{2}-\gamma_{3}} \rho_{02}^* &  \sqrt{(1-\gamma_{1}) (1-\gamma_{2}-\gamma_{3})}\rho_{12}^* &  (1-\gamma_{2}-\gamma_{3}) \rho_{22} \\
\end{array}
\right) \; .
\end{equation}
Resonant Multilevel Amplitude Damping (ReMAD) channels~\cite{ReMAD}, describe relaxation processes of excited to less energetic states in qudit systems with resonant energy levels.  
In the case of a qutrit, the associated map can be expressed as
\begin{equation}\label{eq: channel actionReMAD}
\hspace*{-0.3cm}
\Gamma_{\gamma_1, \gamma_2, \gamma_3}(\hat{\rho})=\left(
\begin{array}{ccc}
\rho_{00} + \gamma_{1} \rho_{11}+\gamma_{3} \rho_{22} &  \sqrt{1-\gamma_{1}} \rho_{01} + \sqrt{\gamma_1\gamma_2}\rho_{12}  &
    \sqrt{1-\gamma_{2}-\gamma_{3}} \rho_{02} \\
\sqrt{1-\gamma_{1}}\rho_{01}^* + \sqrt{\gamma_1\gamma_2}\rho_{12}^*   & (1-\gamma_{1}) \rho_{11}+\gamma_{2} \rho_{22} &
    \sqrt{(1-\gamma_{1}) (1-\gamma_{2}-\gamma_{3})} \rho_{12} \\
 \sqrt{1-\gamma_{2}-\gamma_{3}} \rho_{02}^* &  \sqrt{(1-\gamma_{1}) (1-\gamma_{2}-\gamma_{3})}\rho_{12}^* &  (1-\gamma_{2}-\gamma_{3}) \rho_{22} \\
\end{array}
\right) \; .
\end{equation}
Notice how in the element 01  and 10 of the output density matrix we have now a mixing of the coherences $\rho_{01}$ and $\rho_{12}$, which is absent in MAD channels.
\end{widetext}

\subsection{Capacitance gap for MAD and ReMAD channels}
In panel a) of Fig. \ref{fig:MADworld} we plot a lower bound on the capacitance gap $\Delta C$ in function of the input energy $\mathfrak{e}$ for an instance of a MAD channel $\Phi_{\gamma_1,\gamma_2,\gamma_3}$ and of a ReMAD channel $\Gamma_{\gamma_1,\gamma_2,\gamma_3}$. The lower bound is obtained by numerically calculating the difference between $\ergo^{(1)}_{\rm tot}(\Phi_{\gamma_1,\gamma_2,\gamma_3};\mathfrak{e})$ and $\ergo^{(1)}(\Phi_{\gamma_1,\gamma_2,\gamma_3};\mathfrak{e})$ (respectively for the ReMAD channel the quantity $\ergo^{(1)}_{\rm tot}(\Gamma_{\gamma_1,\gamma_2,\gamma_3}) - \ergo^{(1)}(\Gamma_{\gamma_1,\gamma_2,\gamma_3})$ has been computed). In order to prove that this difference represents a lower bound on the capacitance gap we show in Fig. \ref{fig:MADconcav} that the functions $\ergo^{(1)}_{\rm tot}(\Phi_{\gamma_1,\gamma_2,\gamma_3};\mathfrak{e})$, $\ergo^{(1)}(\Phi_{\gamma_1,\gamma_2,\gamma_3};\mathfrak{e})$, $\ergo^{(1)}_{\rm tot}(\Gamma_{\gamma_1,\gamma_2,\gamma_3});\mathfrak{e}$ and $\ergo^{(1)}(\Gamma_{\gamma_1,\gamma_2,\gamma_3};\mathfrak{e})$ are concave functions of the energy $\mathfrak{e}$ for the values of $\gamma_1$, $\gamma_2$ and $\gamma_3$ reported in the main text, this property implies that 
\begin{eqnarray}
&& \ergo^{(1)}(\Phi_{\gamma_1,\gamma_2,\gamma_3};\mathfrak{e}) = \chi(\Phi_{\gamma_1,\gamma_2,\gamma_3};\mathfrak{e}) \; , \nonumber \\
&& \ergo^{(1)}(\Gamma_{\gamma_1,\gamma_2,\gamma_3};\mathfrak{e}) = \chi(\Gamma_{\gamma_1,\gamma_2,\gamma_3};\mathfrak{e}) \; ,
\end{eqnarray}
and 
\begin{eqnarray}
&& \ergo^{(1)}_{\rm tot}(\Phi_{\gamma_1,\gamma_2,\gamma_3};\mathfrak{e}) = \chi_{\rm tot}(\Phi_{\gamma_1,\gamma_2,\gamma_3};\mathfrak{e}) \; , \nonumber \\
&& \ergo^{(1)}_{\rm tot}(\Gamma_{\gamma_1,\gamma_2,\gamma_3};\mathfrak{e}) = \chi_{\rm tot}(\Gamma_{\gamma_1,\gamma_2,\gamma_3};\mathfrak{e}) \; .
\end{eqnarray}
\begin{figure} 
    \centering
    \includegraphics[width=\linewidth]{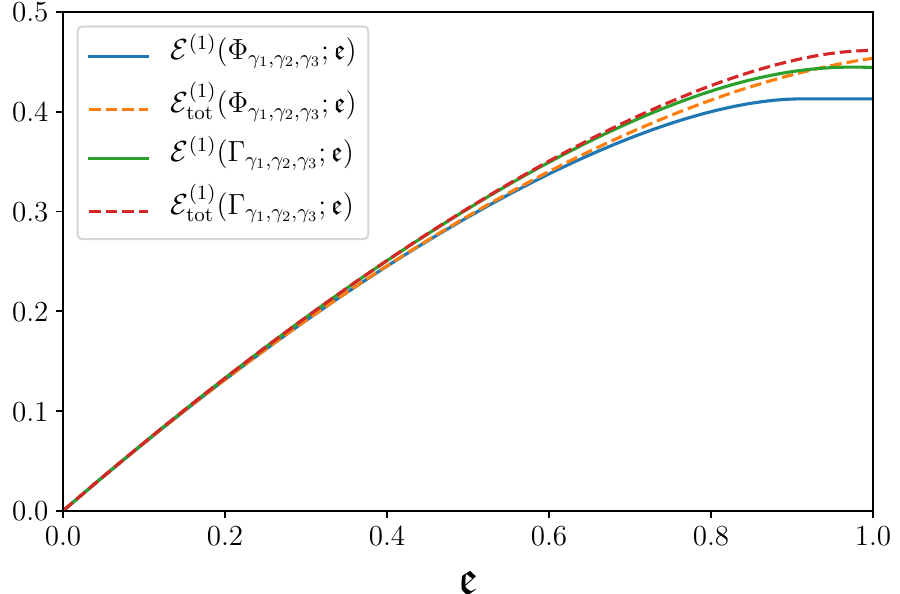}
    \caption{In figure we plot the single-shot optimal output ergotropies $\ergo^{(1)}(\Phi_{\gamma_1,\gamma_2,\gamma_3};\mathfrak{e})$ and $\ergo^{(1)}(\Gamma_{\gamma_1,\gamma_2,\gamma_3};\mathfrak{e})$, and the single-shot output total ergotropies $\ergo^{(1)}_{\rm tot}(\Phi_{\gamma_1,\gamma_2,\gamma_3};\mathfrak{e})$ and $\ergo^{(1)}(\Gamma_{\gamma_1,\gamma_2,\gamma_3};\mathfrak{e})$. We notice that they are all concave functions of the average input energy $\mathfrak{e}$. Here the values of the parameters are $\gamma_1=0.3$, $\gamma_2=0.2$ and $\gamma_3=0.6$.}
    \label{fig:MADconcav} 
\end{figure}
Since for any channel $\Lambda$ and input average energy, $\mathfrak{e}$ as stated in the main text, $\chi(\Lambda;\mathfrak{e})=C_{\rm loc}(\Lambda;\mathfrak{e})$ and $\chi_{\rm tot}(\Lambda;\mathfrak{e}) \leq C_{\rm sep}(\Lambda;\mathfrak{e})$, it is straightforward that for the values of the parameters considered
\begin{eqnarray}
&&\Delta C(\Phi_{\gamma_1,\gamma_2,\gamma_3};\mathfrak{e}) \geq \ergo^{(1)}_{\rm tot}(\Phi_{\gamma_1,\gamma_2,\gamma_3};\mathfrak{e}) - \ergo^{(1)}(\Phi_{\gamma_1,\gamma_2,\gamma_3};\mathfrak{e}) \; , \nonumber \\
&& \Delta C(\Gamma_{\gamma_1,\gamma_2,\gamma_3};\mathfrak{e}) \geq \ergo^{(1)}_{\rm tot}(\Gamma_{\gamma_1,\gamma_2,\gamma_3};\mathfrak{e}) - \ergo^{(1)}(\Gamma_{\gamma_1,\gamma_2,\gamma_3};\mathfrak{e}) \; . \nonumber\\
\end{eqnarray}

\end{document}